\documentclass[10pt]{article}
\usepackage[left=1.8cm, right=3.cm, top=1.5cm, bottom=1.5cm]{geometry}

\usepackage[
    backend=biber,
    style=alphabetic,
    sorting=ynt
]{biblatex}
\usepackage{fullpage}
\usepackage{caption}
\usepackage{subcaption}
\usepackage{amsmath,amsfonts,amssymb}
\usepackage{amsthm}
\usepackage[T1]{fontenc}
\usepackage{etoolbox}
\usepackage{textcomp}
\usepackage{optidef}
\usepackage{graphicx}
\usepackage{paralist}
\usepackage{comment}
\usepackage{bm}
\usepackage{url}
\usepackage{graphics}
\usepackage{multirow,longtable}
\usepackage{rotating}
\usepackage[dvipsnames]{xcolor}
\usepackage{booktabs}
\usepackage{pdflscape}
\usepackage{tikz}
\usepackage{float}
\usepackage[normalem]{ulem}
\usepackage{nicefrac}
\usepackage{color} %,soul}
\usepackage{pgfplots}
\usepackage{xspace}
\usepackage{mathtools}
\usepackage{environ}
\usepackage{lineno} % line numbers

\usepackage{thmtools}
\usepackage{thm-restate}
\usepackage{hyperref}
\usepackage[capitalize,noabbrev]{cleveref}
\usepackage{authblk} %authors + affiliations
\usepackage{subfiles}
\usepackage{bbm}
\usepackage{enumerate}
\usepackage[shortlabels]{enumitem}

\usepackage[linesnumbered,lined,boxed,commentsnumbered, ruled,vlined]{algorithm2e}
\usepackage{algpseudocode}
\usepackage{aligned-overset}
\usepackage[symbol]{footmisc} %footnote with symbol

\usetikzlibrary{arrows.meta, decorations.pathreplacing, shapes, arrows, automata, fit, positioning, calc}

\renewcommand{\geq}{\geqslant}
\renewcommand{\leq}{\leqslant}

\allowdisplaybreaks
\newif\ifhideproofs
\ifhideproofs
\usepackage{environ}
\NewEnviron{hide}{}

\fi

\theoremstyle{plain}
\newtheorem{theorem}{Theorem}
\newtheorem{lemma}[theorem]{Lemma}
\newtheorem{corollary}[theorem]{Corollary}

\newtheorem{definition}[theorem]{Definition}

\theoremstyle{definition}

\theoremstyle{remark}

\newtheorem*{note*}{Note}

\newtheorem*{remark*}{Remark}

\title{A Faster Parametric Search for the Integral Quickest Transshipment Problem}

\author[1]{Mariia Anapolska}
\author[1,*]{Dario van den Boom}
\author[1]{Christina Büsing}
\author[1]{Timo Gersing}

{
\affil[1]{\small Teaching and Research Area Combinatorial Optimization, RWTH Aachen University, Germany}
\affil[*]{Corresponding author: \texttt{van.den.boom@combi.rwth-aachen.de}}
}
\date{}  % Remove date
\begin{document}

\maketitle

\begin{abstract} 
    Algorithms for computing fractional solutions to the quickest transshipment problem have been significantly improved since Hoppe and Tardos first solved the problem in strongly polynomial time.
    For integral solutions, runtime improvements are limited to general progress on submodular function minimization, which is an integral part of Hoppe and Tardos' algorithm.
    Yet, no structural improvements on their algorithm itself have been proposed.
    We replace two central subroutines in the algorithm with methods that require vastly fewer minimizations of submodular functions.
    This improves the state-of-the-art runtime from $ \tilde{\mathcal{O}}(m^4 k^{15}) 
    $ down to $ \tilde{\mathcal{O}}(m^2k^5 + m^4 k^2) $\footnote[2]{We use $ \tilde{\mathcal{O}} $ to suppress polylogarithmic terms.}, where $ k $ is the number of terminals and $ m $ is the number of arcs.
\end{abstract}

\paragraph{Funding}
This work is co-funded by the German research council (DFG) - project number 442047500 - Collaborative Research Center Sparsity and Singular Structures (SFB 1481) and by the German research council (DFG) Research Training Group 2236 UnRAVeL.

\newcommand{\dnet}{\mathcal{N}}

\newcommand{\shat}{\hat{s}}
\newcommand{\scheck}{\check{s}}

\newcommand{\backwards}[1]{{\overset{\leftarrow}{#1}}}
\newcommand{\forwards}[1]{{\overset{\rightarrow}{#1}}}
\newcommand{\residual}[1]{{\overset{\leftrightarrow}{#1}}}

\newcommand{\subpath}[3]{#1_{\mid #2 \text{, } #3}}

\newcommand{\excess}[3]{\text{ex}_{#1}(#2, #3)}
\newcommand{\net}[3]{\text{net}_{#1}(#2, #3)}

\newcommand{\orderto}[1]{(-\infty, #1]_\prec}
\newcommand{\ordertoexc}[1]{(-\infty, #1)_\prec}

\newcommand{\Naux}{\mathfrak{N}}
\newcommand{\saux}{\mathfrak{s}}
\newcommand{\aaux}{\mathfrak{a}}

\newcommand{\chains}{\overset{\leftrightarrow}{\mathcal{P}}}

\newcommand{\dario}[1]{\todo{DA: #1}}
\newcommand{\maria}[1]{\todo[color=yellow]{MA: #1}}

\newcommand{\megiddo}{Megiddo}
\newcommand{\maxAlpha}{\textsc{Maxi\-mize\-Alpha}\xspace}
\newcommand{\minDelta}{\textsc{Mini\-mize\-Delta}\xspace}
\newcommand{\sfm}{\operatorname{SFM}}
\newcommand{\mcf}{\operatorname{MCF}}
\newcommand{\define}{\coloneqq}

\newcommand{\restr}[2]{\tilde{#1}^{#2}}

\section{Introduction}

Network flows over time, also referred to as dynamic flows, extend classical static network flows by a time component. They provide a powerful tool for modeling real-world problems in traffic engineering, building evacuation, and logistics.
Over the last decades, a wide range of optimization problems dealing with flows over time have been studied.
The \emph{maximum flow over time problem} was studied in the seminal work of Ford and Fulkerson~\cite{ford1962flows}, who showed that the problem can be solved in polynomial time by a reduction to the static minimum cost flow problem.
The \emph{quickest flow problem} asks for the minimum time-horizon $ T^* $ such that a provided demand of $ D \in \mathbb{N} $ units of flow can be sent from a single source $ s $ to a single sink $ t $.
While a straightforward approach is to combine an algorithm for maximum flows over time with a parametric search algorithm \cite{burkard1993quickest}, recent results have shown that cost scaling algorithms for minimum cost flows can be modified in order to efficiently compute quickest flows~\cite{lin2014quickest,saho2017cancel}.

The \emph{quickest transshipment problem} generalizes the quickest flow problem by allowing for \emph{supply} and \emph{demand} at multiple \emph{sources} and \emph{sinks}.
It is one of the most fundamental problems in the field of network flows over time and, as recently stated by Skutella~\cite{skutella2023introduction}, ``arguably the most difficult flow over time problem that can still be solved in polynomial time.''
As in the quickest flow problem, our goal is to send the required flow from sources to sinks while simultaneously minimizing the time horizon.

Similar to the quickest flow problem, the quickest transshipment problem can be solved by determining the minimum time horizon via parametric search. 
Using this idea, Hoppe and Tardos~\cite{hoppe2000quickest} showed that the quickest transshipment problem can be solved in strongly polynomial time, that is, their algorithm's runtime is polynomially bounded in the number of nodes $ n $, number of arcs $ m $, and the combined number of sources and sinks $ k $.

Recently, faster algorithms have been developed.
Notably, Schlöter, Tran and Skutella~\cite{schloter2022faster} proposed an algorithm with a time complexity of $ \tilde{\mathcal{O}}(m^2k^5 + m^3 k^3 + m^3n) $.
Unfortunately, these performance improvements come at the expense of fractional solutions, which may be undesirable for applications that do not allow flow particles to be disassembled. 
While some of the results speed up the search for the optimal time horizon, no improvements have yet been proposed for finding integral flows over time.
Hence, the state-of-the-art complexity for the integral quickest transshipment problem remains unchanged at~$ \tilde{\mathcal{O}}(m^4 k^{15}) $~\cite{schloter2022faster}. 

\subsection*{Our Contribution}

We propose improvements to the algorithm by Hoppe and Tardos, reducing the state-of-the-art runtime for the integral quickest transshipment problem from~$ \tilde{\mathcal{O}}(m^4 k^{15}) $ to~$ \tilde{\mathcal{O}}(m^2 k^5 + m^4 k^2) $.
This narrows the gap to the fractional quickest transshipment problem, which can be solved in~$ \tilde{O}(m^2k^5 + m^3 k^3 + m^3n) $ time using the algorithm by Schlöter, Tran and Skutella~\cite{schloter2022faster}.

\section{Preliminaries}
Given a directed graph $ G = (V, A) $ with vertices $ V $ and arcs $ A $, we define a \emph{dynamic network} as a triple $ \dnet = (G, u, \tau) $ with \emph{capacity}~$ u_a \in \mathbb{N} $ and \emph{transit time}~$ \tau_a \in \mathbb{N}_0 $ for each arc~$ a \in A(\dnet) $.
For a given dynamic network $ \dnet $, the set $ V(\dnet) $ denotes the network's nodes, while $ A(\dnet) $ refers to the network's arcs.
Throughout this paper, we denote the number of nodes $ |V(\dnet)| $ by $ n $ and the number of arcs $ |A(\dnet)| $ by $ m $.

As our results do not directly require the definitions of flows over time, we only give a brief introduction and refer to the paper by Skutella \cite{skutella2009introduction} for more details.
A flow over time is a family of functions $ f_a \colon [0, T) \to \{0,\dots,u_a \} $, representing the in-flow rates for each arc~$ a \in A(\dnet) $ for every point in time until the end of the \emph{time horizon}~$ T \in \mathbb{N} $.
The~$ f_a(\theta) $-many flow units entering arc~$ a $ at time~$ \theta \in [0,T) $ arrive at time~$ \theta+\tau_a $ at the end node of $a$.
To satisfy dynamic flow conservation, all flow units entering a non-sink node must immediately be forwarded via an outgoing arc.

For the \emph{dynamic transshipment problem}, we are given a triple $ (\dnet, b, T) $ comprising a dynamic network $ \dnet$, a \emph{balance function} $ b \colon V(\dnet) \to \mathbb{Z} $ with $ \sum_{v \in V(\dnet)} b(v) = 0 $, and a time horizon $ T $.
The balances describe how supply and demand are distributed across the network.
A node with positive balance $ b(v) > 0 $ is a \emph{source}, while a node with negative balance is a \emph{sink}.
Let $ S^+ $ denote the set of sources, $ S^- $ the set of sinks, and $ S = S^+ \cup S^- $ the set of \emph{terminals}.
A dynamic transshipment instance is \emph{feasible} if there exists a flow over time sending the supply from the sources to the sinks such that all demands are satisfied.

\begin{definition} \label{def:out-flow}
    Given a subset of terminals $ X \subseteq S $, the \emph{maximum out-flow $ o(X) $ out of $ X $} is the value of the maximum flow over time from the sources $ S^+ \cap X $ to the sinks $ S^- \setminus X $.
\end{definition}

The central feasibility criterion states that the \emph{net balance} $ b(X) \coloneqq \sum_{v \in X} b(v) $ must not exceed the maximum out-flow $ o(X) $ for every $ X \subseteq S $.

\begin{theorem}[Feasibility Criterion \cite{hoppe2000quickest}] \label{theorem:feasibility}
    The dynamic transshipment instance $ (\dnet, b, T) $ is feasible if and only if $ v(X) \define o(X) - b(X) \geq 0 $ for all $ X \subseteq S $.
\end{theorem}

We call a set $ X \subseteq S $ with $ v(X) < 0 $ a \emph{violated set}.
In order to determine the feasibility of a given dynamic transshipment instance, it suffices to show that no violated set exists.
However, while the value of $ o(X) $, and thus of $ v(X) $, can be computed using the Ford-Fulkerson algorithm for maximum flows over time, avoiding enumeration of all subsets $ X \subseteq S $ is not obvious.
Fortunately, we can utilize the \emph{submodularity} of $ o \colon 2^S \to \mathbb{N}_0 $ and $ v \colon 2^S \to \mathbb{Z} $.

\begin{definition}[Submodular Function]
    A function $ f \colon 2^E \to \mathbb{R} $ over a finite \emph{ground set} $ E $  is a \emph{submodular function} if for all $ X, Y \subseteq E $ it holds that 
    \begin{equation} \label{eq:submodular-function-classic}
        f(X \cup Y) + f(X \cap Y) \leq f(X) + f(Y),
    \end{equation}
    or, equivalently, if for all $ e \in E $ and $ X \subseteq Y \subseteq E \setminus \{ e \} $ it holds that
    \begin{equation}\label{eq:submodular-function}
        f(Y \cup \{ e \}) - f(Y) \leq f(X \cup \{ e \}) - f(X).
    \end{equation}
\end{definition}

Given a submodular function $ f $ over $ E $, we call a set~$ X^* \in \operatorname{argmin}_{X \subseteq E} f(X) $ a \emph{minimizer} of~$f$.
It is well-know that the set of minimizers of a submodular function is closed under union and intersection.
Therefore, there always exists a \emph{minimal minimizer} and a \emph{maximal minimizer}, which are the intersection and union of all minimizers, respectively.

In the context of \emph{submodular function minimization (SFM)}, a submodular function $ f $ is typically provided in form of an \emph{evaluation oracle} with complexity $ \mathcal{O}(\operatorname{EO}) $.
The performance of algorithms is measured in the number of oracle calls required for finding a minimizer.
The fastest strongly polynomial algorithm for submodular function minimization is due to Lee, Sidford and Wong~\cite{lee2015faster} with a runtime of $ \mathcal{O}(k^3 log^2 k \cdot \mathcal{O}(\operatorname{EO}) + k^4 log^{\mathcal{O}(1)} k) $ for~$ k = |E| $.

For our purpose, the evaluation oracle computes a maximum out-flow $ o(X) $ out of terminals $ X \subseteq S $ using the Ford-Fulkerson algorithm for maximum flows over time.
To this end, we introduce a super-source $ s^+ $ and a super-sink $ s^- $ and connect them to sources $ s \in S^+ \cap X $ and sinks $ t \in S^- \setminus X $, respectively, via infinite-capacity, zero-transit arcs. 
Then the maximum out-flow $ o(X) $ is the value of the maximum flow over time from $ s^+ $ to $ s^- $.
This value can be computed using a static min-cost flow in $ \mathcal{O}(m \log n (m + n \log n)) $ or $ \tilde{\mathcal{O}}(m^2) $ time via Orlin's algorithm \cite{orlin1988faster}.
We abbreviate this runtime as $ \mathcal{O}(\mcf(n,m)) $.

Consequently, determining whether a dynamic transshipment instance $ (\dnet, b, T) $ is feasible takes $ \mathcal{O}(k^3 log^2 k \cdot \mcf(n, m) + k^4 log^{\mathcal{O}(1)} k) $ or $ \tilde{\mathcal{O}}(m^2k^3) $ time, where $k = |S|$  is the number of terminals.
To improve readability, we denote the time it takes to check if an instance with $ k $ terminals, $ n $ nodes and $ m $ arcs is feasible by $ \mathcal{O}(\sfm(k, n, m)) $.

\subsection*{The Algorithm by Hoppe and Tardos}

The algorithm by Hoppe and Tardos \cite{hoppe2000quickest} was the first strongly polynomial algorithm computing quickest transshipments and remains the most efficient one for integral solutions.
In the following, we assume that the minimum time horizon $ T^* $ is provided and focus on the algorithm's segment that computes an integral dynamic transshipment.
The algorithm relies on the concept of tight orders.

\begin{definition}[Tight Set and Order]
    A set of terminals $ X \subseteq S $ is called \emph{tight} if $ o(X) = b(X) $ holds. 
    An order $ \prec $ over $ S $  is called \emph{tight} if sets $ \{ t' \in S \mid t' \preceq t \} $ are tight for all $ t \in S $.
\end{definition}

The general idea by Hoppe and Tardos is to construct an equivalent dynamic transshipment instance for which a tight order exists.

\begin{theorem}[Reduction to Lex-Max Flows \cite{hoppe2000quickest}] \label{theorem:tightness}
    Given a dynamic transshipment instance $ (\dnet, b, T) $ with a tight order $ \prec $ over the terminals $ S $, an integral dynamic transshipment satisfying $ b $ can be computed as a lex-max flow over time in $ \mathcal{O}(k \mcf(n, m)) $ time.
\end{theorem}

We refer to \cite{skutella2023introduction} for more details on lex-max flows over time.
Although computing the integral flow is quite efficient, transforming $ (\dnet, b, T) $ into an equivalent instance $ (\dnet', b', T) $ with a tight order is computationally demanding.
We propose improvements to this transformation in \cref{sec:algorithm}. 
Before that, we briefly discuss the approach by Hoppe and Tardos.

First, one adds a new terminal $\scheck$ with $b(\scheck)=b(s)$ for every terminal $s\in S$ and then sets $b(s)=0$.
Note that the set of terminals $S$ now only contains the nodes $\scheck$.
By adding infinite-capacity, zero-transit arcs $(\scheck,s)$ for $ \scheck \in S^+ $ and $ (s, \scheck)$ for $ \scheck \in S^- $, one ensures that the resulting dynamic transshipment instance is equivalent to the original instance (cf.~\cref{fig:alpha-viz}).

To construct an instance admitting a tight order, we iteratively shift supply / demand from terminals $\scheck$ to new terminals $\shat$.
In this paper, we call the $\scheck$ \emph{drained terminals} and the $\shat$ \emph{filled terminals}.
We refer the reader to Hoppe and Tardos \cite{hoppe2000quickest} and \Cref{apx:pseudocode} for a general overview of the algorithm.
We focus on the two subroutines, \maxAlpha and \minDelta, which compute the amount of supply / demand that is shifted from $ \scheck $ to $ \shat $.

Discussing both subroutines, we assume that a feasible dynamic transshipment instance~$ (\dnet, b, T) $ is given, with a set of terminals $ S $ consisting of the drained terminals $ \scheck $, added in the first step, and all filled terminals $ \shat $ that were introduced in previous iterations.
Furthermore, we are given two tight sets $ Q \subset R \subseteq S $ and a drained terminal~$ \scheck \in R \setminus Q $ satisfying~$ o(Q \cup \{ \scheck \}) > b(Q \cup \{ \scheck \}) $.

Note that the instance $ (\dnet, b, T) $, its set of terminals $ S $ as well as $ Q $, $ R $, and $ \scheck $ are the input of our subroutines and vary between iterations.
Throughout this paper, we define~$ \hat{Q} \define Q \cup \{ \shat \} $ and~$ \hat{R} \coloneqq R \cup \{ \shat \} $.
Similar to Hoppe and Tardos, we only discuss the case in which $ \scheck $ is a source, since the treatment of sinks is symmetrical and is sketched in \Cref{apx:symmetry}.

\subsubsection*{Capacity-Parametric Instances}

\begin{figure}[t]
    \centering
    \begin{subfigure}[h]{0.45\textwidth}
    \begin{tikzpicture}[>=stealth',shorten >=1pt, shorten <=0.5pt, auto, node distance = 1.5cm]
        \tikzstyle{every state}=[thick, inner sep=0mm, minimum size=5mm]
    
        \node[state] (s)   {$s$};
        \node[state] (scheck) [left=2cm of s] {$\scheck$};
        \node[] (bal) [below=0.1cm of scheck] {$ b(\scheck) $};

        \node (x) [above of=s, xshift=-1.75cm, yshift=-0.75cm] {};
        \path[->, draw] (x) --node{\small$u_a/\tau_a$}  ++(0.9,0) {};

        \path[->] (scheck) edge node[below, font=\small]{$\infty/0$} (s);
    
        \node[] (p1) [right=1cm of s] {$\dots$};
        \node[] (p2) [above=0.7cm of p1] {$\dots$};
        \node[] (p3) [below=0.7cm of p1] {$\dots$};
    
        \path[->] (s) edge (p1);
        \path[->] (p2) edge (s);
        \path[->] (s) edge (p3);
    \end{tikzpicture}
\end{subfigure}%
\hfill
\begin{subfigure}[h]{0.45\textwidth}
    \begin{tikzpicture}[>=stealth',shorten >=1pt, shorten <=0.5pt, auto, node distance = 1.5cm]
        \tikzstyle{every state}=[thick, inner sep=0mm, minimum size=5mm]
    
        \node[state] (s)   {$s$};
        \node[state] (shat) [left=2cm of s, yshift=0.5cm] {$\shat$};
        \node[state] (scheck) [left=2cm of s, yshift=-0.5cm] {$\scheck$};
        \node[] (bal) [below=0.1cm of scheck] {$ b(\scheck) - \Delta^\alpha $};
        \node[] (bal2) [above=0.1cm of shat] {$ \Delta^\alpha $};

        \path[->] (scheck) edge node[below, font=\small]{$\infty/0$} (s);
        \path[->] (shat) edge node[above, font=\small]{$\alpha/0$} (s);
    
        \node[] (p1) [right=1cm of s] {$\dots$};
        \node[] (p2) [above=0.7cm of p1] {$\dots$};
        \node[] (p3) [below=0.7cm of p1] {$\dots$};
    
        \path[->] (s) edge (p1);
        \path[->] (p2) edge (s);
        \path[->] (s) edge (p3);
    \end{tikzpicture}
\end{subfigure}
    \caption{A dynamic transshipment instance $ (\dnet, b, T) $ after initialization with terminal $\scheck $ (left) and the corresponding $ \alpha $-parametric instance $ (\dnet^\alpha, b^\alpha, T) $ for \maxAlpha (right).}
    \label{fig:alpha-viz}
\end{figure}
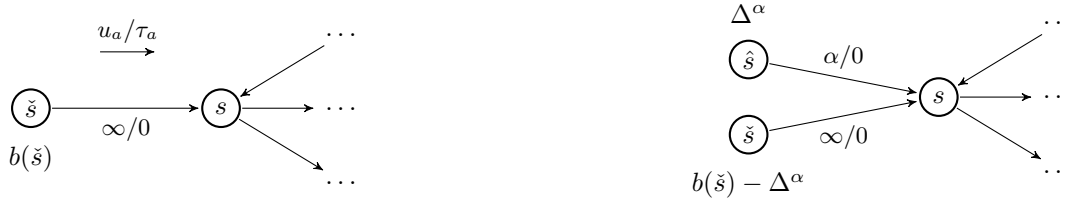

The first subroutine starts with tight sets $ Q $ and $ R $ and a drained source $ \scheck \in R \setminus Q $ for which~$ Q \cup \{ \scheck \} $ is not tight.
Its aim is to reassign as much supply as possible to a new filled source $\hat{s}$.
In doing so, we capacitate the out-flow of $\hat{s}$ such that $\hat{Q}$ is tight and the transformed instance remains feasible.

\begin{definition}[$ \alpha $-Parametric Dynamic Network] \label{def:alpha-parametrized-network}
    Given a dynamic network $ \dnet $, a drained source $ \check{s} \in S^+ $ and a parameter $ \alpha \in \mathbb{N}_0 $, the \emph{$ \alpha $-parametric network} $ \dnet^\alpha $ is constructed by adding a new filled source $ \hat{s} $ and connecting it to $ s $ via an $ \alpha $-capacity, zero-transit arc $ (\hat{s}, s) $.
\end{definition}

Let $ o^\alpha \colon 2^{S\cup\{\shat\}} \to \mathbb{N}_0 $ be the parametric counterpart of the maximum out-flow as in Definition~\ref{def:out-flow} in the parametric network~$ \dnet^\alpha $.
Using this notation, we define \emph{$ \alpha $-parametric dynamic transshipment instances} analogously to Hoppe and Tardos~\cite{hoppe2000quickest}.
The construction of an $ \alpha $-parametric instance is illustrated in \cref{fig:alpha-viz}.

\begin{definition}[$ \alpha $-Parametric Dynamic Transshipment Instance]\label{def:alpha-parametrized-instance}
    Given a feasible dynamic transshipment instance $ (\dnet, b, T) $, two tight sets of terminals $ Q \subset R \subseteq S $, a drained source $ \check{s} \in R \setminus Q $ and a parameter $ \alpha \in \mathbb{N}_0 $, the corresponding \emph{$ \alpha $-parametric dynamic transshipment instance} $ (\dnet^\alpha, b^{\alpha}, T) $ consists of the following components.
    \begin{itemize}
        \item An $ \alpha $-parametric dynamic network $\dnet^\alpha$ as in \cref{def:alpha-parametrized-network}.
        \item An $ \alpha $-parametric balance function $ b^\alpha $ with $ b^\alpha(t) = b(t) $ for all terminals $ t \in S \setminus \{ \check{s} \} $, $ b^\alpha(\hat{s}) = \Delta^\alpha $, and $ b^\alpha(\check{s}) = b(\check{s}) - \Delta^\alpha $, where $ \Delta^\alpha \define o^\alpha(\hat{Q}) - o^\alpha(Q) $.
    \end{itemize}
\end{definition}

We call a parameter value $ \alpha \in \mathbb{N}_0 $ \emph{feasible} if the corresponding $\alpha$-parametric dynamic transshipment instance $(\dnet^\alpha, b^\alpha, T)$ is feasible.
Determining whether a value $ \alpha $ is feasible is equivalent to checking if a \emph{violated} set $ X \subseteq S \cup \{ \shat \} $ exists.
Recall that this can be done by minimizing the parametric submodular function $ v^\alpha(X) = o^\alpha(X) - b^\alpha(X) $.
A subroutine of the algorithm by Hoppe and Tardos finds a maximum feasible parameter value $ \alpha \in \mathbb{N}_0 $.

\begin{definition}[\maxAlpha]
    Given an $\alpha$-parametric dynamic transshipment instance, find the maximum feasible $ \alpha \in \mathbb{N}_0 $.
\end{definition}

We denote the maximum feasible parameter value by $ \alpha^* $.
Hoppe and Tardos concluded that, given a feasibility oracle for $ (\dnet^\alpha, b^\alpha, T) $ taking $ \mathcal{O}(\sfm(k, n, m)) $ time, the value $ \alpha^* $ can be found in $ \mathcal{O}(\log (nU_\text{max}) \cdot \sfm(k, n, m)) $ time, where $ U_\text{max} \define \max_{a \in A(\dnet)} u_a $.
In addition, we can employ \megiddo's parametric search~\cite{megiddo1978combinatorial} to achieve a strongly polynomial runtime of~$ \mathcal{O}(\sfm(k, n, m)^2) $.
We mainly improve upon the strongly polynomial approach.
 
\subsubsection*{Transit-Parametric Instances}

\begin{figure}[t]
    \centering
    \begin{subfigure}[h]{0.45\textwidth}
    \begin{tikzpicture}[>=stealth',shorten >=1pt, shorten <=0.5pt, auto, node distance = 1.5cm]
        \tikzstyle{every state}=[thick, inner sep=0mm, minimum size=5mm]
    
        \node[state] (s)   {$s$};
        \node[state] (scheck) [left=2cm of s] {$\scheck$};
        \node[] (bal) [below=0.1cm of scheck] {$ b(\scheck) $};

        \node (x) [above of=s, xshift=-1.75cm, yshift=-0.75cm] {};
        \path[->, draw] (x) --node{\small$u_a/\tau_a$}  ++(0.9,0) {};

        \path[->] (scheck) edge node[below, font=\small]{$\infty/0$} (s);
    
        \node[] (p1) [right=1cm of s] {$\dots$};
        \node[] (p2) [above=0.7cm of p1] {$\dots$};
        \node[] (p3) [below=0.7cm of p1] {$\dots$};
    
        \path[->] (s) edge (p1);
        \path[->] (p2) edge (s);
        \path[->] (s) edge (p3);
    \end{tikzpicture}
\end{subfigure}%
\hfill
\begin{subfigure}[h]{0.45\textwidth}
    \begin{tikzpicture}[>=stealth',shorten >=1pt, shorten <=0.5pt, auto, node distance = 1.5cm]
        \tikzstyle{every state}=[thick, inner sep=0mm, minimum size=5mm]
    
        \node[state] (s)   {$s$};
        \node[state] (shat) [left=2cm of s, yshift=0.5cm] {$\shat$};
        \node[state] (scheck) [left=2cm of s, yshift=-0.5cm] {$\scheck$};
        \node[] (bal) [below=0.1cm of scheck] {$ b(\scheck) - \Delta^\delta $};
        \node[] (bal2) [above=0.1cm of shat] {$ \Delta^\delta $};

        \path[->] (scheck) edge node[below, font=\small]{$\infty/0$} (s);
        \path[->] (shat) edge node[above, font=\small]{$1/\delta$} (s);
    
        \node[] (p1) [right=1cm of s] {$\dots$};
        \node[] (p2) [above=0.7cm of p1] {$\dots$};
        \node[] (p3) [below=0.7cm of p1] {$\dots$};
    
        \path[->] (s) edge (p1);
        \path[->] (p2) edge (s);
        \path[->] (s) edge (p3);
    \end{tikzpicture}
\end{subfigure}
    \caption{A dynamic transshipment instance $ (\dnet, b, T) $ after initialization with terminal $\scheck $ (left) and the corresponding $ \delta $-parametric instance $ (\dnet^\delta, b^\delta, T) $ for \minDelta (right).}
    \label{fig:delta-viz}
\end{figure}
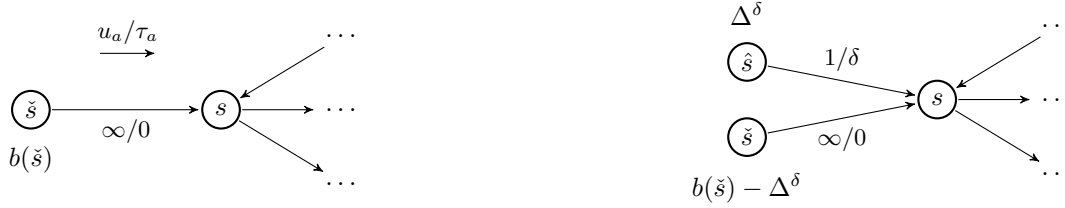

The second subroutine is closely related to \maxAlpha. 
Again, we start with tight sets $ Q $ and $ R $ and a drained source $ \scheck \in R \setminus Q $ for which~$ Q \cup \{ \scheck \} $ is not tight.
The aim is to reassign as much supply as possible to a new filled source $\hat{s}$.
In doing so, we set the transit time for the flow out of $\hat{s}$ such that $\hat{Q}$ is tight and the transformed instance remains feasible.

\begin{definition}[$ \delta $-Parametric Dynamic Network] \label{def:delta-parametrized-network}
    Given a dynamic network $ \dnet $, a drained source $ \check{s} \in S^+ $ and a parameter $ \delta \in \mathbb{N}_0 $, the \emph{$ \delta $-parametric network} $ \dnet^\delta $ is constructed by adding a terminal $ \hat{s} $ and connecting it to $ s $ via a unit-capacity, $ \delta $-transit arc $ (\hat{s}, s) $.
\end{definition}

Analogously to $ o^\alpha $, we denote the maximum out-flow of a subset of terminals $ X \subset S \cup \{ \shat \} $ by~$ o^\delta(X) $.
Next, we combine this parametric variant of our dynamic network with a parametric balance function to form a \emph{$ \delta $-parametric dynamic transshipment instance} (cf. \cref{fig:delta-viz}).

\begin{definition}[$ \delta $-Parametric Dynamic Transshipment Instance] \label{def:delta-parametrized-instance}
    Given a feasible dynamic transshipment instance $ (\dnet, b, T) $, two tight sets of terminals $ Q \subset R \subseteq S $, a drained source $ \scheck \in R \setminus Q $ and a parameter $ \alpha \in \mathbb{N}_0 $, a \emph{$ \delta $-parametric dynamic transshipment instance} $ (\dnet^\delta, b^\delta, T) $ consists of the following components.
    \begin{itemize}
        \item A $ \delta $-parametric dynamic network $\dnet^\delta$ as given in \Cref{def:delta-parametrized-network}.
        \item A $ \delta $-parametric balance function $ b^\delta $ with $ b^\delta(t) = b(t) $ for all terminals $ t \in S \setminus \{ \check{s} \} $, $ b^\delta(\hat{s}) = \Delta^\delta $, and $ b^\delta(\check{s}) = b(\check{s}) - \Delta^\delta $, where $ \Delta^\delta \define o^\delta(\hat{Q}) - o^\delta(Q) $.
    \end{itemize}
\end{definition}

Again, a parameter value $ \delta \in \mathbb{N}_0 $ is \emph{feasible} if the corresponding parametric dynamic transshipment instance $ (\dnet^\delta, b^\delta, T) $ is feasible.
This can be checked by minimizing the parametric submodular function $ v^\delta(X) = o^\delta(X) - b^\delta(X) $.
Again, we call a set $X\subseteq S$ with~$v^\delta(X)<0$ \emph{violated}.
This yields the following parametric search problem.

\begin{definition}[\minDelta]
Given a $\delta$-parametric dynamic transshipment instance, find the minimum feasible $ \delta \in \mathbb{N}_0 $.
\end{definition}

Again, $ \delta^* $ denotes the minimum feasible parameter value.
Hoppe and Tardos showed that the optimal value $ \delta^* $ can be found in $ \mathcal{O}(\log(T) \cdot \sfm(k, n, m)) $ time, or in strongly polynomial time of $ \mathcal{O}(\sfm(k, n, m)^2) $ using the parametric search of \megiddo \cite{megiddo1978combinatorial}.

The algorithm by Hoppe and Tardos performs a total of $ k $ iterations, each of which calls both subroutines \maxAlpha and \minDelta once.
Therefore, the current version of the algorithm by Hoppe and Tardos takes $ \mathcal{O}\left(k \cdot \log (nTU_\text{max}) \cdot \sfm(k, n, m)\right) $ time if both subroutines are implemented with a binary search, while \megiddo's parametric search results in a runtime of $ \mathcal{O}(k \sfm(k, n, m)^2) $.
In the following sections, we improve the runtime of each iteration by introducing better parametric search algorithms for both subroutines.

\section{Restricting the Domains of Violated Sets}\label{sec:restrictions}
Both problems \maxAlpha and \minDelta introduced in the previous section rely on minimizing submodular functions to determine whether a parameter value is feasible.
Even with state-of-the-art algorithms for SFM, this remains a computationally expensive subroutine that scales poorly with the number of terminals in the ground set.

We show that there always exists a violated set $ X \subseteq S $ satisfying $ \hat{Q} \subset X \subset \hat{R} $, which allows us to restrict the ground sets of our parametric submodular functions $ v^\alpha $ and $ v^\delta $.
This provides a practical improvement and establishes the foundation for the following sections.

\begin{lemma} \label{lemma:s_hat-s_check-violator}
    Let $ X \subseteq S \cup \{ \shat \} $ be a violated set for an $ \alpha $-parametric dynamic transshipment instance with respect to a drained source $ \scheck $.
    Then $ \shat \in X $ and $ \scheck \not \in X $.
    The same applies to $ \delta $-parametric dynamic transshipment instances.
\end{lemma}

\begin{proof}
    We only prove the statement for $ \alpha $-parametric instances, as the reasoning is analogous for $ \delta $-parametric instances. 
    Let $ X \subseteq S \cup \{ \hat{s}, \check{s} \} $ be an arbitrary subset of terminals.
    Given that we assume the transshipment instance to be feasible, we have $v(X)\geq 0$.
    Remember that the maximum out-flow $ o^\alpha(X) $ is the value of a maximum flow over time from a super-source $ s^+ $ to a super-sink $ s^- $, where infinite-capacity, zero-transit arcs are used to connect $ s^+ $ to every source in $ S^+ \cap X $ and to connect every sink in $ S^- \setminus X $ to $ s^- $.
    Using this definition, we derive the following observations regarding $ o(X) $:
    \begin{enumerate}[O1]
        \item \label{obs:shat-scheck-1} If $ X $ contains neither $ \shat $ nor $ \scheck $, then maximum out-flow $ o^\alpha(X) $ coincides with the out-flow $ o(X) $ in the original instance.
        \item \label{obs:shat-scheck-2} If $ X $ contains $ \scheck $, then all flow traveling from $ s^+ $ to $ s $ can bypass the $\alpha$-capacity arc $ (\shat, s) $ and move along the infinite-capacity arcs $ (s^+, \scheck) $ and $ (\scheck, s) $ instead. As a consequence, we have $ o^\alpha(X) = o(X) $.
    \end{enumerate}    
    We refer the reader back to \cref{fig:alpha-viz} for a visual intuition.
    Combining these observations with the parametric balances $ b^\alpha(\shat) = \Delta^\alpha $ and $ b^\alpha(\scheck) = b(\scheck) - \Delta^\alpha $, we prove \cref{lemma:s_hat-s_check-violator} by considering the following three complementary cases.
    \begin{enumerate}[(I)]
        \item If $ \hat{s} \in X $ and $ \check{s} \in X $, then it follows that
        \begin{equation*}
            b^\alpha(X) = b^\alpha(X \setminus \{ \shat, \scheck \}) + b^\alpha(\shat) + b^\alpha(\scheck) \overset{\text{Def.~of } b^\alpha}{=} b^\alpha(X \setminus \{ \shat, \scheck \}) + b(\scheck) = b(X \setminus \{ \shat \}) = b(X).
        \end{equation*}
        Together with Observation \ref{obs:shat-scheck-2} we obtain $ v^\alpha(X) = v(X) \geq 0 $, meaning that $ X $ cannot be a violated set.
        \item If $ \hat{s} \not \in X $ and $ \check{s} \in X $, then it follows that $ b^\alpha(X) = b(X) - \Delta^\alpha < b(X) $, which, together with Observation \ref{obs:shat-scheck-2}, implies that $ X $ cannot be a violated set since $ v^\alpha(X) \geq v(X) \geq 0 $.
        \item If $ \hat{s} \not \in X $ and $ \check{s} \not \in X $, then $ b^\alpha(X) = b(X) $ and therefore $ X $ cannot be a violated set as it follows from Observation \ref{obs:shat-scheck-1} that $ v^\alpha(X) = o^\alpha(X) - b^\alpha(X) = o(X) - b(X) = v(X) \geq 0 $.
    \end{enumerate}
    Hence, $ X $ can only be a violated set of $ (\dnet^\alpha, b^\alpha, T) $ if $ \hat{s} \in X $ and $ \check{s} \not \in X $.
\end{proof}

Hoppe and Tardos \cite{hoppe2000quickest} proved the property from \cref{lemma:s_hat-s_check-violator} for the special case of the infeasible parameter value $ \delta^*-1 $ and used it to show that there exists a set $ X $ satisfying $ \hat{Q} \subset X \subset \hat{R} $ which is violated for $ \delta^*-1 $ and tight for $ \delta^* $.
We employ similar arguments to show that this also holds for all infeasible $ \alpha, \delta \in \mathbb{N}_0 $.

\begin{lemma} \label{lemma:alpha-minimizer-within-Q-R}
    Let $ \alpha \in \mathbb{N}_0 $ be an infeasible parameter value for \maxAlpha.
    Then there is a minimizer $ X $ of $ v^\alpha $ with $ \hat{Q} \subset X \subset \hat{R} $.
    Analogously, let $ \delta \in \mathbb{N}_0 $ be an infeasible parameter value for \minDelta.
    Then there is a minimizer $ X $ of $ v^\delta $ with $ \hat{Q} \subset X \subset \hat{R} $.
\end{lemma}

\begin{proof}
    We only prove the statement for \maxAlpha, as the proof for \minDelta is analogous.
    Let $ X^* $ be an arbitrary minimizer of $ v^\alpha $ with $ v^\alpha(X) < 0 $.
    It follows from \cref{lemma:s_hat-s_check-violator} that $ \shat \in X^* $ and $ \scheck \not \in X^* $.
    Next, we show that the set $ X \define Q \cup (X^* \cap \hat{R}) $ is also a minimizer of $ v^\alpha $.
    For this, we analyze the tightness of the sets $ Q $, $ \hat{Q} $, and $ \hat{R} $:
    \begin{itemize}
        \item $ Q $ was chosen to be a tight set for $ (\dnet, b, T) $.
        This also does not change in the parametric instance as $ \shat, \scheck \not \in Q $ and hence $ v^\alpha(Q) = v(Q) = 0$.
        \item $ \hat{Q} $ is tight, since $ b^\alpha(\hat{Q}) = b^\alpha(Q) + \Delta^\alpha \overset{Q \text{ tight}}{=} o^\alpha(Q) +  o^\alpha(\hat{Q}) - o^\alpha(Q) = o^\alpha(\hat{Q}) $.
        \item $ \hat{R} $ is tight because $ R $ is tight and $ \shat, \scheck \in \hat{R} $ directly imply $ o^\alpha(\hat{R}) = o(R) = b(R) = b^\alpha(\hat{R}) $.
    \end{itemize}
    Having shown that $ Q $ and $ \hat{R} $ are tight sets, we study how tight sets and minimizers of~$ v^\alpha $ behave under union and intersection.
    For this purpose, let $ Y \in \{ Q, \hat{R} \} $.
    It follows from submodularity of $ v^\alpha $ that
    \begin{equation} \label{eq:union-intersection-closure}
        v^\alpha(X^* \cup Y) + v^\alpha(X^* \cap Y) \leq v^\alpha(X^*) + v^\alpha(Y) = v^\alpha(X^*) < 0,
    \end{equation}
    and, as a consequence, exactly one of the following properties must hold:
    \begin{enumerate}[(1)]
        \item \label{case:Q-R-tight-1} Both $ X^* \cup Y $ and $ X^* \cap Y $ are violated sets of $ v^\alpha $.
        \item \label{case:Q-R-tight-2} Either $ X^* \cup Y $ or $ X^* \cap Y $ is a minimizer, while the other set is tight.
    \end{enumerate}
    We apply this case distinction to the cases where $ Y = Q $ and $ Y = \hat{R} $:
    \begin{itemize}
        \item Let $ Y = \hat{R} $, then \cref{lemma:s_hat-s_check-violator} states that a violated set cannot contain $ \scheck $, implying that $ v^\alpha(X^* \cup \hat{R}) \geq 0 $ since $ \scheck \in \hat{R} $.
        This means that $ X^* \cap \hat{R} $ is a minimizer of $ v^\alpha $.
        \item Consider $ Y = Q $ and the minimizer $ X^* \cap \hat{R} $. 
        We rule out $ Q \cap (X^* \cap \hat{R}) = Q \cap X^* $ as a violated set since $ \shat \not \in Q \cap X^* $. 
        Therefore, $ Q \cup  (X^* \cap \hat{R}) $ is a minimizer.
    \end{itemize}
    Finally, observe that the minimizer $ X =  Q \cup  (X^* \cap \hat{R}) $ not only satisfies $ \hat{Q} \subseteq X \subseteq \hat{R} $ because~$ \hat{s} \in X^* \cap \hat{R} $ but also~$ \hat{Q} \subset X \subset \hat{R} $ since $ \hat{Q} $ and $ \hat{R} $ are tight.
\end{proof}

\cref{lemma:alpha-minimizer-within-Q-R} allows us to determine feasibility of a parameter value by minimizing the \emph{restricted functions} $ \tilde{v}^\alpha \colon 2^{\hat{R} \setminus \hat{Q}} \to \mathbb{Z} $ and $ \tilde{v}^\delta \colon 2^{\hat{R} \setminus \hat{Q}} \to \mathbb{Z} $ with $ \tilde{v}^\alpha(X) \define v^\alpha(\hat{Q} \cup X) $ and $ \tilde{v}^\delta(X) \define v^\delta(\hat{Q} \cup X) $.
For future use, we define the functions $ \tilde{o}^\alpha $, $ \tilde{o}^\delta $, $ \tilde{b}^\alpha $, and $ \tilde{b}^\delta $ analogously.

This brings us to the main result in this section.

\begin{corollary} \label{corr:faster-feasibility}
    A parameter value $ \alpha \in \mathbb{N}_0 $ is feasible for \maxAlpha if and only if $ \tilde{v}^\alpha(X) \geq 0 $ for every $ X \subseteq \hat{R} \setminus \hat{Q} $. 
    Analogously, a parameter value $ \delta \in \mathbb{N}_0 $ is feasible for \minDelta if and only if $ \tilde{v}^\delta(X) \geq 0 $ for every $ X \subseteq \hat{R} \setminus \hat{Q} $.
\end{corollary}

The obvious advantage of this result is the reduction of the domain of the submodular functions of interest to $ \hat{R} \setminus \hat{Q} $.
This is further accentuated in iterations of Hoppe and Tardos' algorithm, where $ \hat{R} \setminus \hat{Q} $ only comprises very few elements.
However, these restrictions bring further advantages, which we discuss in more detail in the following section.

\section{Strong Map Sequences}

We concluded the previous section with a useful restriction of our submodular functions $ v^\alpha $ and $ v^\delta $ to sets $ X $ with $ \hat{Q} \subset X \subset \hat{R} $, while maintaining the guarantee that our parametric instance is feasible if and only if the minimizer of our restricted function is not violated.
Building on this, we show that both restricted functions satisfy the \emph{strong map property}.

\begin{definition}[Strong Map Property \cite{schloter2018flows}] \label{def:strong-map-property}
    Let $ f_1, f_2 \colon 2^E \to \mathbb{R} $ be two submodular functions defined over the same finite ground set $ E $. We write $ f_1 \sqsupset f_2 $, or $ f_2 \sqsubset f_1 $, if $ X \subseteq Y \subseteq E $ implies
    \begin{equation*}
        f_1(Y) - f_1(X) \leq f_2(Y) - f_2(X).
    \end{equation*} 
    The relation is called a \emph{strong map}. Submodular functions $ f_1, f_2, \dots, f_k $ form a \emph{strong map sequence} if $ f_1 \sqsupset f_2 \sqsupset \dots \sqsupset f_k $.
\end{definition}

Recall that the minimizers of submodular functions are closed under union and intersection, meaning that there exists a unique minimal and maximal minimizer for every submodular function. 
A result by Topkis \cite{topkis_minimizing_1978} relates the minimal and maximal minimizers of functions that form strong map sequences.

\begin{lemma}[Minimizers for Strong Map Sequences \cite{schloter2018flows}] \label{lemma:strong-map-sequences-minimizers}
    Let $ f_1, f_2 \colon 2^E \to \mathbb{R} $ be two submodular functions over the same ground set $ E $ with $ f_1 \sqsubset f_2 $. Denote by $ X_1^\text{min} $ and $ X_1^\text{max} $ the minimal and maximal minimizer of $ f_1 $, respectively, while $ X_2^\text{min} $ and $ X_2^\text{max} $ are the minimal and maximal minimizer of $ f_2 $, respectively. 
    Then $ X_1^\text{min} \subseteq X_2^\text{min} $ and $ X_1^\text{max} \subseteq X_2^\text{max} $.
\end{lemma}

\Cref{lemma:strong-map-sequences-minimizers} already gives an intuition for how the strong map property can be used for the parametric search problems \maxAlpha and \minDelta:
each choice of~$ \alpha $ or~$ \delta $ gives us a new submodular function, and, if these functions form strong map sequences, the strong map property implies that we can only encounter at most $ |S| $ distinct minimal minimizers for a monotonous sequence of parameter values.

Unfortunately, neither $ v^\alpha $ nor $ v^\delta $ satisfy the strong map property as is. 
This changes, however, when considering the restricted functions $ \tilde{v}^\alpha $ and $ \tilde{v}^\delta $ defined in Section~\ref{sec:restrictions}.
Recall that these functions are defined over the ground set $ \hat{R} \setminus \hat{Q} $, and that their definition ensures that~$ \shat $ is implicitly added, while~$ \scheck $ is excluded.

\begin{lemma} \label{lemma:throttle-strong-map}
    Let $ 0 \leq \alpha \leq \alpha' $. Then $ \restr{v}{\alpha} \sqsubset \restr{v}{\alpha'} $.
\end{lemma}

\begin{proof}
    Our argument is structured as follows: We first prove that $ \restr{o}{\alpha} \sqsubset \restr{o}{\alpha+1} $ holds for all~$ \alpha \in \mathbb{N}_0 $ and then use this result to prove that $ \tilde{v}^\alpha $ also forms a strong map sequence with~$ \restr{v}{\alpha} \sqsubset \restr{v}{\alpha+1} $. 
    The general claim $ \restr{v}{\alpha} \sqsubset \restr{v}{\alpha'} $ then immediately follows by transitivity of $ \leq $.

    \begin{figure}[hbt]
        \centering
        \resizebox{0.5\textwidth}{!}{
            \begin{tikzpicture}[>=stealth',shorten >=1pt, shorten <=0.5pt, auto, node distance = 1.5cm]
    \tikzstyle{every state}=[thick, inner sep=0mm, minimum size=5mm]

    \node[state] (s)   {$s$};
    \node[state] (shat) [left=2cm of s] {$\shat$};
    \node[state] (sfrac) [above=0.5cm of shat] {$\mathfrak{s}$};
    \node[state, gray] (scheck) [below=0.5cm of shat] {$\scheck$};

    \node[] (p1) [right=1cm of s] {$\dots$};
    \node[] (p2) [above=0.7cm of p1] {$\dots$};
    \node[] (p3) [below=0.7cm of p1] {$\dots$};

    \path[->] (s) edge (p1);
    \path[->] (p2) edge (s);
    \path[->] (s) edge (p3);

    \path[->] (sfrac) edge node[above right, font=\small]{$1/0$} (s);
    \path[->] (shat) edge node[below, font=\small]{$\alpha/0$} (s);
    \path[->, gray] (scheck) edge node[below right, font=\small]{$\infty/0$} (s);

    \node (x) [right=2 of s] {};
    \path[->, draw] (x) --node{\small$u_a/\tau_a$}  ++(0.9,0) {};   
\end{tikzpicture}
        }
        \caption{The auxiliary network $ \Naux $ for \maxAlpha.}
        \label{fig:alpha-aux-network}
    \end{figure}
    
    Let $ \alpha \in \mathbb{N}_0 $ be arbitrary but fixed. 
    We construct an auxiliary network $ \Naux $ with nodes $ V(\Naux) \define V(\dnet^\alpha) \cup \{ \mathfrak{s} \} $ and arcs~$ A(\Naux) \define A(\dnet^\alpha) \cup \{ \mathfrak{a} \define (\mathfrak{s}, s) \} $, where $ s $ is the source that was replaced by $ \scheck $ in the first phase of the algorithm. 
    For the new arc, we set~$ u_\mathfrak{a} = 1 $ and~$ \tau_\mathfrak{a} = 0 $.
    An example network for a given source $ \shat $ is depicted in \cref{fig:alpha-aux-network}.
    
    Let $ \mathfrak{o}(X) $ be the max out-flow function for network $\Naux$ restricted to the domain~$ \hat{R} \setminus \hat{Q} \cup \{ \saux \} $. 
    Notice how adding $ \saux $ to a set of terminals $ X \subseteq \hat{R} \setminus \hat{Q} \cup \{ \saux \} $ with $ \hat{s} \in X $ and $ \check{s} \not \in X $ is equivalent to increasing $ \alpha $ by one, as flow cannot bypass the arcs~$ (\mathfrak{s}, s) $ and~$ (\shat, s) $ through~$ (\scheck, s) $.
    Formally, it holds for every set $ X $ with $ X \subseteq \hat{R} \setminus \hat{Q} $ that
    \begin{enumerate}
        \item \label{item:obs2} $ \mathfrak{o}(X) = \restr{o}{\alpha}(X) $, and
        \item \label{item:obs1} $ \mathfrak{o}(X \cup \{ \saux \}) = \restr{o}{\alpha+1}(X) $.
    \end{enumerate}
    Clearly, the function $ \mathfrak{o} $ is also submodular.
    Hence, given two sets $ X \subseteq Y \subset \hat{R} \setminus \hat{Q} $, the definition of submodularity in \cref{eq:submodular-function} can be restated as
    \begin{equation} \label{eq:submod-modified}
        \mathfrak{o}(Y \cup \{ \mathfrak{s} \}) - \mathfrak{o}(X \cup \{ \mathfrak{s} \}) \leq \mathfrak{o}(Y) - \mathfrak{o}(X).
    \end{equation}
    Combining \cref{eq:submod-modified} with all our previous observations, it follows that the sets~$ X $ and~$ Y $ satisfy
    \begin{equation*}
        \begin{split}
            \restr{o}{\alpha+1}(Y) - \restr{o}{\alpha+1}(X)   \overset{\text{Obs. \ref{item:obs1}}}&{=} \mathfrak{o}(Y \cup \{ \saux \}) - \mathfrak{o}(X \cup \{ \saux \}) \\
            \overset{\text{Eq. \eqref{eq:submod-modified}}}&{\leq} \mathfrak{o}(Y) - \mathfrak{o}(X) \\
            \overset{\text{Obs. \ref{item:obs2}}}&{=} \restr{o}{\alpha}(Y) - \restr{o}{\alpha}(X).
        \end{split}
    \end{equation*}
    Hence, we have shown that $ \restr{o}{\alpha} \sqsubset \restr{o}{\alpha+1} $ holds for every parameter value $ \alpha $.
    Finally, recall our definition of $ \tilde{v}^{\alpha} $ as $ \tilde{v}^{\alpha}(X) = \tilde{o}^{\alpha}(X) - \tilde{b}^{\alpha}(X) $ for any $ X \subseteq \hat{R} \setminus \hat{Q} $.
    Due to the strong map property of $ \tilde{o}^\alpha $, we get for all $ X \subseteq Y \subseteq \hat{R} \setminus \hat{Q} $ that
    \begin{equation*}
        \begin{split}
            \restr{v}{\alpha+1}(Y) - \restr{v}{\alpha}(Y) &= \restr{o}{\alpha+1}(Y) - \restr{o}{\alpha}(Y) + \tilde{b}^{\alpha}(Y) - \tilde{b}^{\alpha+1}(Y) \\
            &= \restr{o}{\alpha+1}(Y) - \restr{o}{\alpha}(Y) + \Delta^\alpha - \Delta^{\alpha+1} \\
            &\leq \tilde{o}^{\alpha+1}(X) - \tilde{o}^\alpha(X) + \Delta^\alpha - \Delta^{\alpha+1} \\
            &= \restr{o}{\alpha+1}(X) - \restr{o}{\alpha}(X) + \tilde{b}^{\alpha}(X) - \tilde{b}^{\alpha+1}(X) \\
            &= \restr{v}{\alpha+1}(X) - \restr{v}{\alpha}(X), \\            
        \end{split}
    \end{equation*}
    therefore proving that $ \tilde{v}^\alpha \sqsubset \tilde{v}^{\alpha+1} $.
    Thus, by transitivity of $ \leq $, our claim $ \tilde{v}^\alpha \sqsubset \tilde{v}^{\alpha'} $ holds. 
\end{proof}

Next, we utilize a similar argument to prove that $ \tilde{v}^\delta $ forms a strong map sequence. 
Here, we rely on our specific definition of the $ \delta $-parametric network $ \dnet^\delta $.

\begin{lemma} \label{lemma:delay-strong-map}
    Let $ 0 \leq \delta' \leq \delta $. Then $ \restr{v}{\delta} \sqsubset \restr{v}{\delta'} $.
\end{lemma}

\begin{proof} 
    The general approach is identical to the proof of Theorem \ref{lemma:throttle-strong-map}.
    That is, we show that $ \tilde{o}^\delta $ forms a strong map sequence with $ \tilde{o}^{\delta-1} \sqsupset \tilde{o}^{\delta} $ for $ \delta \geq 1 $.
    Afterwards, we transfer this result to $ \tilde{v}^{\delta-1} $ and $ \tilde{v}^\delta $.
    The general claim then directly follows from the transitivity of $ \leq $.

    \begin{figure}[hbt]
        \centering
        \resizebox{0.7\textwidth}{!}{
            \begin{tikzpicture}[>=stealth',shorten >=1pt, shorten <=0.5pt, auto, node distance = 1.5cm]
    \tikzstyle{every state}=[thick, inner sep=0mm, minimum size=5mm]

    \node[state] (s)   {$s$};
    \node[state] (sfrac) [above left of=s, xshift=-1.5cm] {$\mathfrak{s}$};
    \node[state] (shat) [left=2cm of sfrac] {$\shat$};
    \node[state, gray] (scheck) [below left of=s, xshift=-1.5cm] {$\scheck$};

    \node[] (p1) [right=1cm of s] {$\dots$};
    \node[] (p2) [above=0.7cm of p1] {$\dots$};
    \node[] (p3) [below=0.7cm of p1] {$\dots$};

    \path[->] (s) edge (p1);
    \path[->] (p2) edge (s);
    \path[->] (s) edge (p3);

    \path[->] (shat) edge node[above, font=\small]{$1/1$} (sfrac);
    \path[->] (sfrac) edge node[above, font=\small]{$1/\delta-1$} (s);
    \path[->, gray] (scheck) edge node[below, font=\small]{$\infty/0$} (s);

    \node (x) [right=2.5cm of s] {};
    \path[->, draw] (x) --node{\small$u_a/\tau_a$}  ++(0.9,0) {};
    
\end{tikzpicture}
        }
        \caption{The auxiliary network $ \Naux $ for \minDelta.}
        \label{fig:delta-aux-network}
    \end{figure}
    
    Again, let $ \delta \in \mathbb{N} $ be arbitrary but fixed. 
    We construct an auxiliary network $ \mathfrak{N} $ with nodes~$ V(\Naux) \define V(\dnet^\delta) \cup \{ \saux \} $ and arcs~$ A(\Naux) \define (A(\dnet^\delta) \setminus \{(\shat, s)\}) \cup \{ \mathfrak{a} \define (\saux, s), \hat{a} \define (\shat, \saux) \} $, $ s $ is the source that was replaced by $ \scheck $ in the first step of the algorithm.
    Additionally, let $ u_{\hat{a}} = u_{\aaux }= 1 $, $ \tau_{\aaux} = \delta - 1 $ and $ \tau_{\hat{a}} = 1 $.
    An example for a source $ \shat $ is given in \cref{fig:delta-aux-network}.
    
    Again, $ \mathfrak{o}(X) $ denotes the maximum out-flow function for this network restricted to the domain $ \hat{R} \setminus \hat{Q} \cup \{ \saux \} $.
    In this construction, we have replaced the arc $ (\hat{s}, s) $ by a path consisting of arcs $ (\shat, \saux) $ and $ (\saux, s) $ with a combined transit time of $ \delta $.
    As a consequence, the maximum out-flow $ \mathfrak{o}(X) $ remains unchanged for sets $ X \subseteq \hat{R} \setminus \hat{Q} $.
    If $ \saux \in X $ and $ \shat \in X $, on the other hand, both sources compete for access to the arc $ (\saux, s) $.
    Then, the flow out of $ X $ is maximized by sending all flow from $ \saux $.
    Formally, it holds for every set $ X \subseteq \hat{R} \setminus \hat{Q} $ that 
    \begin{enumerate}
        \item \label{item:obs2} $ \mathfrak{o}(X) = \restr{o}{\delta}(X) $, and
        \item \label{item:obs1} $ \mathfrak{o}(X \cup \{ \saux \}) = \restr{o}{\delta-1}(X) $
    \end{enumerate}
    These observations plus the second definition of submodularity from \cref{eq:submodular-function} imply for all sets $ X \subseteq Y \subseteq \hat{R} \setminus \hat{Q} $ that 
    \begin{equation*}
        \begin{split}
            \restr{o}{\delta-1}(Y) - \restr{o}{\delta-1}(X) \overset{\text{Obs. \ref{item:obs1}}}&{=} \mathfrak{o}(Y \cup \{ \saux \}) - \mathfrak{o}(X \cup \{ \saux \}) \\
            \overset{\text{Eq.~\eqref{eq:submodular-function}}}&{\leq} \mathfrak{o}(Y) - \mathfrak{o}(X) \\
            \overset{\text{Obs. \ref{item:obs2}}}&{=} \restr{o}{\delta}(Y) - \restr{o}{\delta}(X).
        \end{split}
    \end{equation*}
    Therefore, the function $ \tilde{o}^\delta $ forms a strong map sequence with $ \tilde{o}^{\delta-1} \sqsupset \tilde{o}^{\delta} $.
    Finally, by the same arguments as in the proof of \cref{lemma:throttle-strong-map}, the parametric function $ \tilde{v}^\delta $ forms a strong map sequence with $ \tilde{v}^\delta \sqsupset \tilde{v}^{\delta'} $ as claimed.
\end{proof}

In the following section, we use the strong map property and the resulting nested sequences of minimizers to construct new parametric search algorithms for $ v^\alpha $ and $ v^\delta $ . 

\section{An Improved Parametric Search} \label{sec:algorithm}
In this section, we adapt Schlöter's \cite{schloter2018flows} parametric search algorithm for the single-source or single-sink quickest transshipment problem to \maxAlpha and \minDelta.
We start by showing that $ \tilde{v}^\alpha $ and $ \tilde{v}^\delta $ are monotonic in their respective parameters. 

\begin{lemma} \label{lemma:throttle-delay-monotonic}
    Given a set of terminals $ X \subseteq \hat{R} \setminus \hat{Q} $, both maps $ \alpha \mapsto \tilde{v}^\alpha(X) $ and $ \delta \mapsto \tilde{v}^\delta(X) $ are monotonic in the parameters $ \alpha $ and $ \delta $, respectively.
    That is, we have $ \tilde{v}^\alpha(X) \geq \tilde{v}^{\alpha+1}(X) $ and $ \tilde{v}^\delta(X) \leq \tilde{v}^{\delta+1}(X) $ for all $ \alpha, \delta \in \mathbb{N}_0 $.
\end{lemma}
\begin{proof}
    Let $ \alpha, \delta \in \mathbb{N}_0 $ be arbitrary but fixed parameter values.  
    We set $Z \coloneqq X \cup\hat{Q}$ and decompose $ \tilde{v}^\alpha(X) $ (and, analogously, $ \tilde{v}^\delta(X) $) as
    \begin{align}\label{eq:alpha-decomposition}
        \tilde{v}^\alpha(X) &= \tilde{o}^\alpha(X) - \tilde{b}^\alpha(X) \\\nonumber
                            &= \tilde{o}^\alpha(X) - \big(b(Z \setminus \{ \shat \}) + \Delta^\alpha\big) \\\nonumber
                            \overset{\text{Def.~\ref{def:alpha-parametrized-instance}}}&{=} \tilde{o}^\alpha(X) - \big(b(Z \setminus \{ \shat \}) + o^\alpha(\hat{Q}) - o^\alpha(Q)\big) \\\nonumber
                            &= \tilde{o}^\alpha(X) - \big(b(Z \setminus \{ \shat \}) + \tilde{o}^\alpha(\emptyset) - o^\alpha(Q)\big) \\\nonumber
                            &= \tilde{o}^\alpha(X) - \tilde{o}^\alpha(\emptyset) + \big(o^\alpha(Q) - b(Z \setminus \{\shat \})\big).\nonumber
    \end{align}
    Observe that we can treat the term $ o^\alpha(Q) - b(Z \setminus \{\shat \}) $ as constant: since $ \shat, \scheck \not \in Q $, the out-flow $ o^\alpha(Q) $ is identical for all $ \alpha \in \mathbb{N}_0 $ (or $ \delta \in \mathbb{N}_0 $).
    Similarly, the value of~$ b(Z \setminus \{ \shat \}) $ does not depend on $ \alpha $ or $ \delta$.
    
    We now prove monotonicity using this simplification.    
    For \maxAlpha, the monotonicity $ \tilde{v}^\alpha(X) \geq \tilde{v}^{\alpha+1}(X) $ holds if and only if~$ \tilde{o}^{\alpha}(X) - \tilde{o}^\alpha(\emptyset) \geq \tilde{o}^{\alpha+1}(X) - \tilde{o}^{\alpha+1}(\emptyset) $ is true, which follows directly from the strong map property $ \tilde{o}^\alpha \sqsubset \tilde{o}^{\alpha+1} $ shown in the proof of \cref{lemma:throttle-strong-map}.
    Similarly, replacing $ \alpha $ with $ \delta $ in \cref{eq:alpha-decomposition}, it follows from $ \tilde{o}^\delta \sqsupset \tilde{o}^{\delta+1} $ that $ \tilde{o}^{\delta}(X) - \tilde{o}^\delta(\emptyset) \geq \tilde{o}^{\delta+1}(X) - \tilde{o}^{\delta+1}(\emptyset) $ and thus $ \tilde{v}^\delta(X) \leq \tilde{v}^{\delta+1}(X) $ holds.
\end{proof}

The monotonicity and strong map property of $ \tilde{v}^\alpha $ and $ \tilde{v}^\delta $ allow us to introduce new parametric search algorithms for \maxAlpha~(cf.~Algorithm~\ref{alg:alpha-parametric-search}) and \minDelta~(cf.~Algorithm~\ref{alg:delta-parametric-search}).
In their core, they follow a rather simple approach similar to algorithms by Schlöter, Skutella and Tran~\cite{schloter2022faster, schloter2018flows}:~we start with an infeasible parameter value $\alpha_1$ or $\delta_1$ and a corresponding minimizer $ X_1 $ of $ \tilde{v}^{\alpha_1} $ or $\tilde{v}^{\delta_1}$.
Possible initial values are~$ \alpha_1 \coloneqq \alpha_\text{max} = nU_\text{max}$ and~$ \delta_1 \coloneqq 0 $ \cite{hoppe2000quickest}.
Next, we alternate between two steps \emph{jump} and \emph{check} until the current minimizer~$ X_i $ is no longer violated.
In the \emph{jump} step, we compute the largest parameter value~$ \alpha_{i+1} $ or the smallest parameter value $ \delta_{i+1} $ such that the previous minimizer~$ X_i $ is no longer violated.
Afterwards, the \emph{check} step finds a minimizer~$ X_{i+1} $ for the new value~$ \alpha_{i+1} $ and~$ \delta_{i+1} $.

\begin{figure}[b]
    \centering
    \begin{algorithm}[H]
    \label{alg:alpha-parametric-search}
        \SetAlgoLined
        \caption{Parametric Search for \maxAlpha}
        
        \KwData{Tight sets $ \hat{Q} \subset \hat{R} \subseteq S \cup \{ \shat \}$, submodular function $ \tilde{v}^\alpha \colon 2^{\hat{R} \setminus \hat{Q}} \to \mathbb{Z} $, infeasible upper bound $ \alpha_\text{max}$}
        \KwResult{Maximum feasible $ \alpha \in \mathbb{N}_0 $}            
        $ \alpha_1 \gets \alpha_\text{max} $, 
        $ X_1 \gets $ Minimizer of $ \tilde{v}^{\alpha_1} $  \\
        $ i \gets 1 $ \\
        \While{$ \tilde{v}^{\alpha_i}(X_i) < 0 $}{
            $ \alpha_{i+1} \gets $ Maximum $ \alpha \geq 0 $ with $ \tilde{v}^{\alpha}(X_i) \geq 0 $ \Comment{Jump step} \\
            $ X_{i+1} \gets $ Minimizer $ X \subset X_i $ of $ \tilde{v}^{\alpha_{i+1}} $ \Comment{Check step} \\
            $ i \gets i + 1 $ \\
        }
        \Return{$ \alpha_i $}
    \end{algorithm}
\end{figure}

\begin{figure}[t]
    \centering
    \begin{algorithm}[H]
        \label{alg:delta-parametric-search}
        \SetAlgoLined
        \caption{Parametric Search for \minDelta}
        
        \KwData{Tight sets $ \hat{Q} \subset \hat{R} \subseteq S \cup \{ \shat \} $, submodular function $ \tilde{v}^\delta \colon 2^{\hat{R} \setminus \hat{Q}} \to \mathbb{Z} $}
        \KwResult{Minimum feasible $ \delta \in \mathbb{N} $}            
        $ \delta_1 \gets 0 $,
        $ X_1 \gets $ Minimizer of $ \tilde{v}^{\delta_1} $  \\
        $ i \gets 1 $ \\
        \While{$ \tilde{v}^{\delta_i}(X_i) < 0 $}{
            $ \delta_{i+1} \gets $ Minimum $ \delta \geq 0 $ with $ \tilde{v}^{\delta}(X_i) \geq 0 $ \Comment{Jump step} \\
            $ X_{i+1} \gets $ Minimizer $ X \subset X_i $ of $ \tilde{v}^{\delta_{i+1}} $ \Comment{Check step} \\
            $ i \gets i + 1 $\\
        }
        \Return{$ \delta_i $}
    \end{algorithm}
\end{figure}

In the check step, we use the results of the previous sections in order to restrict the search for minimizers to $ X_{i+1} \subset X_i $.
We will show in the following that this restriction is actually feasible.
Note that the restriction not only reduces the search space in each iteration to $ |X_i| $ terminals, but also limits \maxAlpha and \minDelta to at most $ |S| $ iterations, since the length of the chain $ X_1 \supset X_2 \supset \dots $ is limited by $ |\hat{R} \setminus \hat{Q}| \leq |S| $.

\begin{theorem} \label{theorem:correctness}
    Algorithm~\ref{alg:alpha-parametric-search} computes the maximum feasible $ \alpha^* \in \mathbb{N}_0 $ for \maxAlpha;
    Algorithm~\ref{alg:delta-parametric-search} computes the minimum feasible $\delta^* \in \mathbb{N}_0$ for \minDelta.
    Both algorithms terminate after at most $ |S| $ iterations of the while loop.
\end{theorem}

\begin{proof}
    We first show by induction that if $ \alpha_{i} $ is infeasible, then there exists a minimizer $ X_{i} $ of~$ \tilde{v}^{\alpha_{i}} $ with $ X_i \subset X_{i-1} $, where $ X_0 = \hat{R} \setminus \hat{Q} $.
    
    By definition,~$ X_1 $ is a minimizer of $ \tilde{v}^{\alpha_{1}} $.
    Let~$\alpha_i$ and~$ \alpha_{i+1} $ be infeasible, and assume that~$ X_{i} $ is a minimizer of $ \tilde{v}^{\alpha_{i}} $.
    We have $ \alpha_{i+1} < \alpha_{i} $, since otherwise $ \tilde{v}^{\alpha_{i+1}}(X_{i}) \overset{\text{Lem.~\ref{lemma:throttle-delay-monotonic}}}{\leq} \tilde{v}^{\alpha_{i}}(X_{i}) < 0 $ contradicts the choice of $ \alpha_{i+1} $ as feasible for $ X_i $.
    Together with \cref{lemma:throttle-strong-map}, this implies the relation~$ \tilde{v}^{\alpha_i} \sqsupset \tilde{v}^{\alpha_{i+1}} $.
    Therefore, we have $ X_{i+1}^\text{min} \subseteq X_{i}^\text{min} \subseteq X_{i} $, with $ X^\text{min}_j $ being the minimal minimizer of $ \tilde{v}^{\alpha_j}$ for $ j\in \{i,i+1\}$.
    Since $\alpha_{i+1}$  is infeasible with $\tilde{v}^{\alpha_{i+1}}(X^\text{min}_{i+1}) <0 $, and, by definition, $\tilde{v}^{\alpha_{i+1}}(X_i) \geq 0$, we even have $ X^\text{min}_{i+1} \subset X_{i} $.
    Hence, $ X_{i+1}^\text{min} $ is a feasible choice for~$ X_{i+1} $ in Algorithm~\ref{alg:alpha-parametric-search}.

    Recall that Algorithm~\ref{alg:alpha-parametric-search} terminates after at most $ |S| $ iterations.
    Let $ \alpha_{i^*} $ be the value returned by the algorithm, and let $ X_{i^*} $ be the minimizer generated in the final iteration $i^*$.
    If $ \alpha_{i^*} $ were not feasible, then $ X_{i^*} $ would be a minimizer with $v^{\alpha_{i^*}}(X_{i^*}) < 0$.
    However, since the algorithm terminates, it follows that $ \alpha_{i^*} $ is feasible.
    
    To see that $\alpha_{i^*}$ is maximum, recall that~$\alpha_1 = \alpha_\text{max}$ is infeasible, and thus the jump step is executed at least once.
    By construction in the jump, the value $ \alpha_{i^*} $ is maximum such that the previous minimizer $ X_{i^*-1} $ is no longer violated.
    Hence, for any $\alpha> \alpha_{i^*}$, we have $ v^{\alpha}(X_{i^*-1}) < 0$, so $ \alpha $ is infeasible.
    Therefore, $ \alpha_{i^*} $ is the optimal solution to \maxAlpha.
    
    The proof for Algorithm~\ref{alg:delta-parametric-search} is analogous.
\end{proof}

The remainder of the section is devoted to the runtime analysis. 
Our main improvement is due to the upper bound $ k = |S| $ on the number of iterations the algorithms execute. 

\begin{theorem} \label{theorem:runtime-maximize-alpha}
    Algorithm~\ref{alg:alpha-parametric-search} can be implemented to solve \maxAlpha in strongly polynomial time of $ \mathcal{O}(k [\sfm(k, n, m) + \mcf(n, m)^2]) $ and in weakly polynomial time of $ \mathcal{O}(k[\sfm(k, n, m) + \log (nU_\text{max}) \cdot \mcf(n, m)]) $.
\end{theorem}

\begin{proof}
    The runtime is determined by the two main steps performed in each iteration:
    \begin{itemize}
        \item Jump can be implemented using \megiddo's parametric search \cite{megiddo1978combinatorial} or binary search over the range $ [0, nU_\text{max}]$ in conjunction with a minimum cost flow algorithm.
        The former results in a runtime of $ \mathcal{O}(\mcf(n, m)^2) $, the latter in $ \mathcal{O}(\log (n U_\text{max}) \cdot \mcf(n, m)) $.
        \item Check minimizes the submodular function $ v^\alpha $ on the restricted domain.
        In the worst case, this takes $ \mathcal{O}(\sfm(k, n, m)) $ time.
    \end{itemize}
    Together with the upper bound of $ k $ on the iterations of the while loop, the final runtime is $ \mathcal{O}(k [\sfm(k, n, m) + \mcf(n, m)^2]) $ or $ \mathcal{O}(k [\sfm(k, n, m) + \log (nU_\text{max}) \cdot \mcf(n, m)]) $.
\end{proof}

We obtain an analogous runtime for Algorithm~\ref{alg:delta-parametric-search}.

\begin{theorem} \label{theorem:runtime-minimize-delta}
    Algorithm~\ref{alg:delta-parametric-search} can be implemented to solve \minDelta in strongly polynomial time of $ \mathcal{O}(k [\sfm(k, n, m) + \mcf(n, m)^2]) $ and in weakly polynomial time of $ \mathcal{O}(k[\sfm(k, n, m) + \log (T) \cdot \mcf(n, m)]) $.
\end{theorem}

\begin{proof}
    The proof is analogous to that of \Cref{theorem:runtime-maximize-alpha}.
    The binary search for the minimum feasible value of $\delta$ is done on the range $[0,T]$, which yields the different logarithmic term.
\end{proof}

We conclude this section with an improved runtime for computing integral dynamic transshipments compared to the state-of-the-art runtime of $ \tilde{\mathcal{O}}(m^4 k^{15}) $.

\begin{theorem}\label{thm:totalRuntime}
    Given a dynamic transshipment instance $ (\dnet, b, T) $, an integral quickest transshipment can be computed in $ \tilde{\mathcal{O}}(m^2k^5 + m^4k^2) $ time.
\end{theorem}

\begin{proof}
    Hoppe and Tardos~\cite{hoppe2000quickest} already proved that their algorithm for dynamic transshipment terminates after $ \mathcal{O}(k) $ iterations, each of which consists of one call of Algorithms~\ref{alg:alpha-parametric-search} and~\ref{alg:delta-parametric-search}.
    Given the strongly polynomial runtime of $ \mathcal{O}\big(k (\sfm(k, n, m) + \mcf(n, m)^2)\big) $ for both subroutines, we obtain a worst-case complexity of $ \mathcal{O}(k^2 (\sfm(k, n, m) + \mcf(n, m)^2)) $.
    
    Suppressing the polylogarithmic terms, we have a time complexity of $ \tilde{\mathcal{O}}(m^2 k^3) $ for $ \mathcal{O}(\sfm(k, n, m)) $ and of $\tilde{\mathcal{O}}(m^2)$ for $\mathcal{O}(\mcf(n,m))$.
    Overall, we obtain an improved runtime of $ \tilde{\mathcal{O}}(m^2 k^5 + m^4 k^3) $ time for the integral dynamic transshipment problem.

    In order to compute a quickest transshipment, we have to determine the minimum time horizon first.
    This can be done by the method by Schlöter, Tran and Skutella~\cite{schloter2022faster} in $\tilde{\mathcal{O}}(m^2 k^5 + m^3 k^3 + m^3n)$ time.
    This runtime is dominated by that of our algorithm for computing the corresponding transshipment: since we assume that the network is connected, we have $m\geq n \geq k$, and thus $\tilde{\mathcal{O}}(m^3n) \subset \tilde{\mathcal{O}}(m^4)$ and $\tilde{\mathcal{O}}(m^3 k^3) \subset \tilde{\mathcal{O}}(m^4 k^2)$.
    Overall, the time required to compute a quickest integral transshipment is
    \[
        \tilde{\mathcal{O}}\big((m^2 k^5 + m^3 k^3 + m^3n) + (m^2 k^5 + m^4 k^2)\big) = \tilde{\mathcal{O}}(m^2 k^5 + m^4 k^2 + m^4) = \tilde{\mathcal{O}}(m^2 k^5 + m^4 k^2),
    \]
    which shows \cref{thm:totalRuntime}.
\end{proof}

\section{Conclusion and Outlook}
In this paper, we propose an improved version of the algorithm by Hoppe and Tardos for the integral quickest transshipment problem.
Our approach is based on more efficient parametric search algorithms using the strong map property and yields a substantial reduction of the runtime from~$ \tilde{\mathcal{O}}(m^4 k^{15}) $ to~$ \tilde{\mathcal{O}}(m^2 k^5 + m^4 k^2) $.

Our findings open room for ensuing research.
In particular, the restrictions of submodular functions to suitable domains introduced in this paper may provide even better bounds on the runtime for our algorithms.
Furthermore, we see potential improvements to the jump steps in Algorithms~\ref{alg:alpha-parametric-search} and \ref{alg:delta-parametric-search} that are currently based on \megiddo's parametric search and contribute a factor of $ \tilde{\mathcal{O}}(m^4) $ to the runtime.
This factor constitutes the remaining gap in runtime between the integral and the fractional quickest transshipment problem.
In order to close this gap, future studies may focus on adapting parametric minimum cost flow algorithms akin to the algorithms by Lin and Jaillet~\cite{lin2014quickest} and Saho and Shigeno~\cite{saho2017cancel}.

\printbibliography

\appendix
\newpage
\section{Algorithm of Hoppe and Tardos for the Quickest Transshipment Problem}\label{apx:pseudocode}

\begin{figure}[ht]
    \centering
    \begin{algorithm}[H]
    \label{alg:hoppe-tardos}
        \SetAlgoLined
        \caption{Algorithm by Hoppe and Tardos \cite{hoppe2000quickest}}
        
        \KwData{A feasible dynamic transshipment instance $ (\dnet, b, T) $ with terminals $ S $}
        \KwResult{Modified instance $ (\dnet', b', T) $ with terminals $ S' $ and tight order $ \prec $ over $ S' $}
        $ S' \gets \{ \check{s} \mid s \in S \} $ \\
        $ b(\scheck) \gets b(s) $ for every $ s \in S $ \\
        $ V(\dnet') \gets V(\dnet) \cup S' $ \\
        $ \widetilde{A} \gets \{ (\check{s}, s) \mid s \in S^+ \} \cup \{ (s, \check{s}) \mid s \in S^- \} $ \\
        $ A(\dnet') \gets A(\dnet) \cup \widetilde{A} $ \\
        $ u_a = \infty $ for all $ a \in \widetilde{A} $ \\
        $ \tau_a = 0 $ for all $ a \in \widetilde{A} $ \\
        $ \mathcal{C} \gets $ Initial chain of thight sets $ (\emptyset, S') $\\

        \While{$ |\mathcal{C}| < |S| + 1 $}{
            $ Q, R \gets $ Neighboring sets in $ \mathcal{C} $ with $ |R \setminus Q| > 1 $ \\
            $ \check{s} \gets $ Arbitrary first terminal $ \check{s} \in R \setminus Q $ \\
            \eIf{$ \check{s} $ is source}{
                \eIf{$ Q \cup \{ \check{s} \}$ is tight}{
                    $ \mathcal{C} \gets \mathcal{C} $ with  $ Q \cup \{ \check{s} \} $ inserted between $ Q $ and $ R $ \\ 
                }{
                    $ \alpha^* \gets \maxAlpha(\dnet', b', T, Q, R, \check{s}) $ \\
                    $ (\dnet', b', T) \gets (\dnet^{\alpha^*}, b^{\alpha^*}, T) $ for $ \check{s} $ and new source $ \hat{s} $ \\
                    $ Q' \gets Q \cup \{ \hat{s} \} $ \\
                    $ \delta^* \gets \minDelta(\dnet', b', T, Q', R, \check{s}) $ \\
                    $ (\dnet', b', T) \gets (\dnet^{\delta^*}, b^{\delta^*}, T) $ for $ \check{s} $ and new source $ \hat{s} $ \\
                    $ Q'' \gets Q' \cup \{ \hat{s} \} $ \\
                    $ W \gets $ Minimizer of $ v^{\delta^*-1} $ \\
                    $ \mathcal{C} \gets \mathcal{C} $ with $ Q', Q'', Q'' \cup \{ W \cap R \} $ inserted between $ Q $ and $ R $ \\
                    Replace $ R $ in $ \mathcal{C} $ with $ R \cup Q'' $ and every $ X \in \mathcal{C} $ after $ R $ with $ Q'' \cup X $ \\
                }
            }{
                \tcp{Symmetrical to sources...}
            }
        }
        \Return{$ (\dnet', b', T) $, $ S' $, $ s \prec s' $ if there exist $ X, Y \in \mathcal{C} $ with $ s \in X $ and $ s' \in Y \setminus X $}
    \end{algorithm}
\end{figure}

\section{Symmetrical Treatment of Sinks}\label{apx:symmetry}

This section discusses how the results of this paper can be applied to the case in which~$ \scheck $ is a sink.
Although all cases are symmetrical, the exact execution of proofs may not be completely clear without further study.
Most of the differences are due to the asymmetrical definition of $ o(X) $ for sources and sinks:
while a source $ s \in S^+ \cap X$ may send flow and thus positively contributes to the out-flow~$ o(X) $, a sink~$ t \in S^-\cap X $ may not receive flow, and thus negatively contributes to~$o(X)$.

The first difference is that subroutines \maxAlpha and \minDelta are executed only if~$ R \setminus \{ \scheck \} $ is not tight, as otherwise~$ R \setminus \{ \scheck \} $ is a sought tight set.

Furthermore, the definitions of $ (\dnet^\alpha, b^\alpha, T) $ and $ (\dnet^\delta, b^\delta, T) $ are slightly modified: given a drained sink~$ \scheck \in R \setminus Q $ connected to a sink $ s $ from the original instance~$ (\dnet, b, T) $, a new sink~$ \hat{s} $ is introduced and connected to~$ s $ via an arc~$ (s, \shat) $.
The arc's capacity and transit time are identical to those in the source case.
However, the main difference is the definition of $ \Delta^\alpha $ and $ \Delta^\delta $ as $ \Delta^\alpha \coloneqq o^\alpha(R \cup \{ \shat \}) - o^\alpha(R) $ and $ \Delta^\delta \coloneqq o^\delta(R \cup \{ \shat \}) - o^\delta(R) $, respectively.

Based on these modified definitions, the following property can be proved analogously to \cref{lemma:s_hat-s_check-violator}.

\begin{lemma} \label{lemma:s_hat-s_check-violator-symmetric-sink}
    Let $ X \subseteq S \cup \{ \shat \} $ be a violated set for an $ \alpha $-parametric dynamic transshipment instance with respect to a drained sink $ \scheck $.
    Then $ \shat \not \in X $ and $ \scheck \in X $.
    The same applies to $ \delta $-parametric dynamic transshipment instances.
\end{lemma}

Again, given a minimizer $ X^* $, we can construct a new minimizer $ X \coloneqq Q \cup (R \cap X^*) $ satisfying $ Q \subset X \subset R $.
This means that $ \tilde{v}^\alpha $ and $ \tilde{v}^\delta $ are both defined on the domain $ R \setminus Q $.

The strong map properties for the restricted functions become reversed.
\begin{lemma} \label{lemma:throttle-strong-map-sink}
    Let $ 0 \leq \alpha \leq \alpha' $. Then $ \restr{v}{\alpha} \sqsupset \restr{v}{\alpha'} $.
\end{lemma}

\begin{lemma} \label{lemma:delay-strong-map-sink}
    Let $ 0 \leq \delta' \leq \delta $. Then $ \restr{v}{\delta} \sqsupset \restr{v}{\delta'} $.
\end{lemma}

From this point on, algorithms solving \maxAlpha and \minDelta for sinks can be derived analogously to Algorithms \ref{alg:alpha-parametric-search} and \ref{alg:delta-parametric-search} .

\end{document}